\documentclass[pdftex,11pt,a4paper]{article}
\usepackage{amsmath,amssymb,amsfonts,amsthm}
\usepackage{algorithmicx, algorithm,algpseudocode}
\usepackage{longtable}
\usepackage{enumerate}
\usepackage{graphicx}
\usepackage{algpseudocode}
\usepackage{algorithm}
\usepackage{stmaryrd}
\usepackage{float}
\usepackage[font=small,labelfont=bf]{caption}
\usepackage[left=3cm, right=3cm, top=2.5cm,  bottom=2.5cm ]{geometry}
\newtheorem{thm}{Theorem}
\newtheorem{lem}{Lemma}
\newtheorem{prp}{Proposition}
\newtheorem{cor}{Corollary}
\newtheorem{defn}{Definition}
\newtheorem{example}{Example}
\newtheorem{rem}{Remark}
\newtheorem{fact}{\bf Fact}

\newcommand{\F}{\mathbb{F}}

\newcommand{\C}{\mathbb{C}}
\usepackage{enumerate}
\usepackage{hyperref}
\hypersetup
{
colorlinks   = true, 
urlcolor     = red, 
linkcolor    = blue, 
citecolor   = blue 
}
\usepackage{xcolor}

\newcommand{\splitatcommas}[1]{%
\begingroup
\begingroup\lccode`~=`, \lowercase{\endgroup
\edef~{\mathchar\the\mathcode`, \penalty0 \noexpand\hspace{0pt plus 1em}}%
}\mathcode`,="8000 #1%
\endgroup
}
\DeclareMathOperator{\Diag}{Diag} 

\usepackage[
n,
operators,
advantage,
sets,
adversary,
landau,
probability,
notions,	
logic,
ff,
mm,
primitives,
events,
complexity,
asymptotics,
keys]{cryptocode}
\date{}

\begin{document}
\title{A Systematic Construction Approach for All $4\times 4$ Involutory MDS Matrices}
\author{Yogesh Kumar\footnote{Scientific Analysis Group, DRDO, Metcalfe House Complex, Delhi-110054\newline email: \texttt{adhana.yogesh@gmail.com}},
P.R.Mishra\footnote{Scientific Analysis Group, DRDO, Metcalfe House Complex, Delhi-110054\newline email: \texttt{prasanna.r.mishra@gmail.com}},
Susanta Samanta\footnote{R. C. Bose Centre for Cryptology and Security, Indian Statistical Institute, Kolkata-700108\newline email: \texttt{susanta.math94@gmail.com}}, 
Atul Gaur\footnote{Department of Mathematics, University of Delhi, Delhi-110007 \newline email: \texttt{gaursatul@gmail.com}}}
\maketitle

\begin{abstract}\noindent
Maximum distance separable (MDS) matrices play a crucial role not only in coding theory but also in the design of block ciphers and hash functions. Of particular interest are involutory MDS matrices, which facilitate the use of a single circuit for both encryption and decryption in hardware implementations. In this article, we present several characterizations of involutory MDS matrices of even order. Additionally, we introduce a new matrix form for obtaining all involutory MDS matrices of even order and compare it with other matrix forms available in the literature. We then propose a technique to systematically construct all $4 \times 4$ involutory MDS matrices over a finite field $\mathbb{F}_{2^m}$. This method significantly reduces the search space by focusing on involutory MDS class representative matrices, leading to the generation of all such matrices within a substantially smaller set compared to considering all $4 \times 4$ involutory matrices. Specifically, our approach involves searching for these representative matrices within a set of cardinality $(2^m-1)^5$. Through this method, we provide an explicit enumeration of the total number of $4 \times 4$ involutory MDS matrices over $\mathbb{F}_{2^m}$ for $m=3,4,\ldots,8$.\\\\

\end{abstract}

\section{Introduction}
The concepts of confusion and diffusion~\cite{shan} play a crucial role in the design of symmetric key cryptographic primitives. In general, the linear transformation is used to achieve the diffusion property, which can be represented as a matrix. Additionally, it is common practice to include maximum distance separable (MDS) matrices in the linear transformations to achieve optimal diffusion. One application is in the MixColumns operation of AES~\cite{AES}. Moreover, MDS matrices find extensive application in numerous other cryptographic primitives, including stream ciphers such as MUGI~\cite{MUGI} and hash functions like Maelstrom~\cite{MAELSTROM}, Gr$\phi$stl~\cite{GROSTL}, and PHOTON~\cite{PHOTON}. Consequently, the efficacy of MDS matrices within diffusion layers is widely acknowledged, leading to the proposal of various techniques for their design.

Generally, three techniques are employed for generating MDS matrices. The first technique, referred to as direct construction, utilizes algebraic methods to produce an MDS matrix of any order without the necessity of search procedures. The second method involves search-based construction. The third approach integrates the strategies mentioned above into a hybrid methodology. Direct constructions typically derive from Cauchy and Vandermonde matrices, along with their extensions. Examples from the literature include references~\cite{Gupta2023direct, kc2, LACAN2003, Roth1989, sdm}.

Although direct construction methods allow for the possibility of obtaining MDS matrices of any order, there is no assurance of achieving a matrix with the most efficient hardware area, even for smaller sizes. The second technique, referred to as the search-based technique, is currently the only method that can give us the best possible MDS matrix in terms of area. Nevertheless, this methodology is only viable in situations where the field size is not too large and the matrix size is small. When it comes to search techniques, there are different matrix structures to consider. These include circulant, left-circulant, Hadamard, and Toeplitz matrices. A considerable amount of progress has been made in this area, as demonstrated in \cite{GR15,CYCLICM,psa,XORM,skf}.

In the third method, referred to as the hybrid approach, a representative MDS matrix is usually obtained using a search method. Afterwards, this matrix is used to generate several MDS matrices. Following a concise overview of prior research in this field, we will now discuss the motivations of this paper. The idea of a representative matrix for involutory MDS matrices first appeared in \cite{gmt}, presenting a matrix form that generates all $3 \times 3$ representative involutory MDS matrices with two parameters. However, it is worth mentioning that in that paper, the authors did not use the hybrid method involving representative matrices. Instead, they provided an explicit form of a $3 \times 3$ involutory MDS matrix over $\mathbb{F}_{2^m}$, from which all $3 \times 3$ involutory MDS matrices can be derived by varying the parameters over $\mathbb{F}_{2^m}$ with some conditions. The authors in \cite{psa} introduced a hybrid method to effectively produce MDS matrices of size $2^n\times 2^n$. Their method presents the GHadamard matrix, a new matrix form that may be employed to create new (involutory) MDS matrices from a Hadamard (involutory) MDS matrix as a representative matrix. In \cite{Saka}, a method was presented to produce new $n\times n$ MDS matrices that are isomorphic to pre-existing ones. The authors in~\cite{Kumar_MDS2024} demonstrate a hybrid approach for generating all MDS matrices, which results in a search space of $2^{9m}$ to generate all $4\times 4$ involutory MDS matrices. Whereas, according to~\cite{Yang2021}, the search space for generating all $4\times 4$ involutory MDS matrices amounts to $2^{8m}$. Furthermore, a recent study~\cite{Tuncay23} proposes a hybrid method to generate all $4\times 4$ involutory MDS matrices over $\mathbb{F}_{2^m}$, reducing the search space to $(2^m-1)^8$ to identify all $4\times 4$ representative involutory MDS matrices. 

Also, it is important to note that the enumeration of all $4\times 4$ involutory MDS matrices in these papers is limited to finite fields of small order. This motivates us to explore the theoretical aspects of effectively generating all $4\times 4$ involutory MDS matrices over $\mathbb{F}_{2^m}$ and reducing the search space to $(2^m-1)^5$. We present a hybrid method to construct all $4 \times 4$ involutory MDS matrices over $\mathbb {F }_{2^m}$ and demonstrate how our method significantly reduces the search space. Initially, we partition the set of all $4 \times 4$ involutory MDS matrices into equivalence classes and identify representative matrices for each class. Utilizing these representative matrices, we efficiently generate all involutory MDS matrices within the same class through direct computation. We illustrate that these involutory MDS class representative matrices of order $4$ can be found by examining only a small subset of all $4 \times 4$ involutory matrices over $\mathbb{F}_{2^m}$. 

Additionally, we provide some characterizations of involutory MDS matrices of even order. Furthermore, we propose a new matrix form for obtaining all involutory MDS matrices of even order and compare it with other matrix forms available in the literature.\\

\noindent The paper is structured as follows. Section~\ref{Sec:Definiotion} briefly discusses the mathematical background and notations used in the article. In Section~\ref{Section_IMDS_even}, we provide some characterizations of involutory MDS matrices of even order. Section~\ref{Sec:proposed_method} presents the theory underlying the generation of all $4 \times 4$ involutory MDS matrices over $\mathbb{F}_{2^m}$. The algorithm to generate all $4\times 4$ involutory MDS matrices and counting all $4 \times 4$ involutory MDS matrices over some finite fields is detailed in Section~\ref{Sec:Experimental_result}. Finally, we conclude in Section~\ref{Sec:Conclusion}.

\section{Mathematical Preliminaries}~\label{Sec:Definiotion}
In this section, we discuss some definitions and mathematical preliminaries that are important in our context. Let $\F_{2^m}$ be a finite field of order $2^m$, where $m$ is a positive integer. Let $\F_{2^m}^*$ denote the multiplicative group of the finite field $\F_{2^m}$.

Let $M$ be any $n\times n$ matrix over $\mathbb{F}_{2^m}$ and let $|M|$ denote its determinant. A matrix $D$ of order $n$ is diagonal if $(D){i,j}=0$ for $i\neq j$. Using $d_i = (D){i,i}$, we represent the diagonal matrix $D$ as $\Diag(d_1, d_2, \ldots, d_n)$. The determinant of $D$ is $|D| = \prod_{i=1}^{n}{d_i}$. Hence, $D$ is non-singular over $\mathbb{F}_{2^m}$ if and only if $d_i \neq 0$ for $1 \leq i \leq n$. 

In this paper, we denote the identity matrix of order $n$ as $I_n$. Let $\mathcal{L}_{2n}$ be the set of $2n \times 2n$ matrices over $\F_{2^m}$ of type 
$\begin{pmatrix}
A &B\\
C &D
\end{pmatrix}$,
where $A, B, C, D$ are $n \times n$ non-singular matrices. 

An MDS matrix is utilized in various cryptographic primitives as a diffusion layer. The notion of an MDS matrix originates from coding theory, particularly from the domain of maximum distance separable (MDS) codes. An $[n, k, d]$ code is considered MDS if it satisfies the singleton bound $d = n-k + 1$.

\begin{thm}~\cite[page 321]{FJ77} \label{thm:1}
An $[n, k, d]$ code $C$ with generator matrix $G = [ I ~|~ M ]$, where $M$ is a $k \times ( n - k )$ matrix, is MDS if and only if every square sub-matrix (formed from any $i$ rows and any $i$ columns, for any $i = 1, 2,\ldots, min \{k, n - k \}$) of $M$ is non-singular.
\end{thm}

\begin{defn}
A matrix $M$ of order $n$ is said to be an MDS matrix if $[I~|~M]$ is a generator matrix of a $[2n,n]$ MDS code.
\end{defn}

Another way to define an MDS matrix is as follows:

\begin{fact}
A square matrix $M$ is an MDS matrix if and only if every square sub-matrix of $M$ is non-singular. 
\end{fact}

\noindent Implementing involutory diffusion matrices is more advantageous since it enables using a single module for both the encryption and decryption processes.

\begin{defn}
An involutory matrix is a square matrix $M$ of order $n$ that is self-invertible, or equivalently, satisfies $M^2 = I_n$.
\end{defn}

\noindent One of the fundamental row operations on matrices involves multiplying a row or column of a matrix by a non-zero scalar. The MDS property remains unchanged under such operations. Therefore, we have the following result concerning MDS matrices.


\begin{lem}~\cite[Theorem 4.7]{GR15} \label{lem:1}
Let $M$ be an involutory MDS matrix, then for any non-singular diagonal matrix $D$, $D^{-1}MD$ will also be an involutory MDS matrix.
\end{lem}

\section{Construction of Involutory MDS Matrices}\label{Section_IMDS_even}
In this section, we explore some characterizations of involutory MDS matrices for even orders. There exists a direct correspondence between the set of $n \times n$ nilpotent matrices of index $2$ over $\F_{2^m}$ and the set of $n\times n$ involutory matrices. It is straightforward to demonstrate that each mapping from $M \mapsto M+I_n$ establishes a bijection between these sets. This correspondence yields valuable insights into identifying involutory matrices.

Rather than focusing on involutory matrices directly, we can examine nilpotent matrices of index $2$. The relationship between involutory matrices and nilpotent matrices of index $2$ is detailed in the following proposition.

\begin{prp} \label{prp:1}
Let $M$ be an $n\times n$ involutory matrix of even order over $\mathbb{F}_{2^m}$, then $(M+I_n)$ is a nilpotent matrix of index $2$ over $\mathbb{F}_{2^m}$, and its rank can be at most $\frac{n}{2}$.
\end{prp}
\begin{proof}
We have
\[(M+I_n)^2=(M^2+I_n)=0.\]
Thus, the index of $(M+I_n)$ is 2. Let's say $N=M+I_n$.

Next, we demonstrate that the rank of N cannot be more than $\frac{n}{2}$.
Let's define the null space and range space of $N$ as $N(N)$ and $R(N)$, respectively. Consequently, $rank(N)=\dim(R(N))$. Let $y$ be an arbitrary element in the set $R(N)$. Then a vector $x$ exists such that $y=Nx$. Now, $Ny=N(Nx)=N^2x=0$. It means $y \in N(N)$. We now have  $R(N)\subset N(N)$. It implies
	
\begin{equation}\label{prp:1_eq:1}
    \dim(R(N))\leq \dim(N(N)). 
\end{equation}

From rank-nullity theorem	
\begin{equation}\label{prp:1_eq:2}
    \dim (R(N))+\dim(N(N))=n.
\end{equation}
(\ref{prp:1_eq:1}) and (\ref{prp:1_eq:2}) together imply that $rank(N)=\dim(R(N))\leq \frac{n}{2}$.
Therefore, rank of $N$ or $(M+I_n)$ can not be more than $\frac{n}{2}$.
\end{proof}
\noindent If $M$ is an involutory MDS matrix in $\mathcal{L}_{2n}$~\footnote{$\mathcal{L}_{2n}$ denotes the set of $2n \times 2n$ matrices over $\F_{2^m}$ of type 
$\begin{pmatrix}
A &B\\
C &D
\end{pmatrix}$, 
where $A, B, C, D$ are $n \times n$ non-singular matrices.}, then, the next theorem demonstrates that the rank of $M+I_{2n}$ is exactly half of the order of $M$. Here, $I_{2n}$ is a $2n\times 2n$ identity matrix.

\begin{thm}\label{thm:2}
Let $M$ be an involutory MDS matrix in $\mathcal{L}_{2n}$. Then, rank$(M+I_{2n})=n$.
\end{thm}

\begin{proof}
    Let $M=(m_{ij})$ and $M+I_{2n}=N=(n_{ij})$. Then, we have
    \begin{equation}\label{thm:2_eq:1}
    m_{ij}=
    \begin{cases}
    n_{ij}\text{ if }i\neq j\\
    n_{ij}+1 \text{ otherwise }.
    \end{cases}
\end{equation}

According to Proposition \ref{prp:1},  $N$ is a nilpotent matrix of index $2$ and rank$(N)\leq\frac{2n}{2}=n$. There are two possible outcomes: either rank$(N)=n$ or rank$(N)<n$. If possible, assume that rank$(N))\neq n$. Then, we must have rank$(N)<n$. In this scenario, all $n\times n$ minors of $N$ are zero.

Take the $n \times n$ sub-matrix $W=(w_{ij})$ obtained from $N$ by the last $n$ rows and the first $n$ columns from it. The determinant of $W$ must be zero. It is evident from (\ref{thm:2_eq:1}) that $W$ is also a sub-matrix of $M$ whose determinant is zero. This is a contradiction because $M$ is an MDS matrix. As a result, we have  rank$(N)=$ rank$(M+I_{2n})=n$.
\end{proof}

\noindent The following theorem discusses a specific form of an involutory matrix $M$ in $\mathcal{L}_{2n}$.

\begin{thm}\label{thm:3}
Let $M \in \mathcal{L}_{2n}$. Then $M$ is involutory if and only if $M$ can be written as follows:
\begin{equation}\label{Eqn_IMDS_even}
    \begin{aligned}
        M &=
        \begin{pmatrix}
            PC &PCP\\
            C &CP
        \end{pmatrix}+I_{2n},
    \end{aligned}
\end{equation}
where $P$ and $C$ are $n\times n$ non-singular matrices over $\mathbb{F}_{2^m}$.
\end{thm}

\begin{proof}
Suppose that $M$ is an involutory matrix. Let $M=N+I_{2n}$. So $N$ is a nilpotent matrix of index $2$. Since $M\in\mathcal{L}_{2n}$, it has the following form: 
\begin{equation*}
M=\begin{pmatrix}
A' &B'\\
C' &D'
\end{pmatrix},
\end{equation*}
where $A',B',C'$, and $D'$ are $n\times n$ non-singular matrices. 
Then
\begin{equation*}
N=\begin{pmatrix}
A'+I_n &B'\\
C' &D'+I_n
\end{pmatrix}
=
\begin{pmatrix}
A &B\\
C &D
\end{pmatrix},
\end{equation*}
where $A=A'+I_n,B=B',C=C'$ and $D=D'+I_n$.\\

Since $C=C'$ is non-singular, we have rank$(N)\geq n$. Furthermore, since $N$ is a nilpotent matrix of index $2$ then its rank cannot exceed $n$. As a result, $rank(N)=n$ is required. Because of this, rows of $A$ can be expressed as a linear combination of rows of $C$. In this instance, a $n \times n$ non-singular matrix $P$ exists such that $A=PC$. Similarly, a linear combination of columns of $C$ can be used to write columns of $D$. It implies that a non-singular matrix $Q$ exists such that $D=CQ$. 
Then, there is
\begin{equation}\label{thm:3_eq:1}
N=\begin{pmatrix}
PC &B\\
C &CQ
\end{pmatrix}.
\end{equation}
Let $O$ be a zero matrix of order $n$.
Pre-multiply $N$ by 
$\begin{pmatrix}
P^{-1} &I_n\\
O &I_n
\end{pmatrix}$ and post-multiply by  
$\begin{pmatrix}
I_n &I_n\\
O &Q^{-1}
\end{pmatrix}$
get the following:
\begin{eqnarray*}
&&\begin{pmatrix}
P^{-1} &I_n\\
O &I_n
\end{pmatrix}N
\begin{pmatrix}
I_n &I_n\\
O &Q^{-1}
\end{pmatrix}\\
&=&\begin{pmatrix}
P^{-1} &I_n\\
O &I_n
\end{pmatrix}
\begin{pmatrix}
PC &B\\
C &CQ
\end{pmatrix}
\begin{pmatrix}
I_n &I_n\\
O &Q^{-1}
\end{pmatrix}\\
&=&\begin{pmatrix}
O &P^{-1}B+CQ\\
C &CQ
\end{pmatrix}
\begin{pmatrix}
I_n &I_n\\
O &Q^{-1}
\end{pmatrix}\\	
&=&\begin{pmatrix}
O &P^{-1}BQ^{-1}+C\\
C &O
\end{pmatrix}.
\end{eqnarray*}
Given that $N$ has the rank $n$, we must have $P^{-1}BQ^{-1}+C=O$ either or $B=PCQ$.
By substituting the value of $B$ in (\ref{thm:3_eq:1}) we get
\begin{equation*}
N=\begin{pmatrix}
PC &PCQ\\
C &CQ
\end{pmatrix}.
\end{equation*}
In addition, since $N$ is a nilpotent matrix of index 2, we have $N^2=0$. To put it another way, i.e.,
\begin{equation*}
\begin{pmatrix}
PCPC+PCQC &PCPCQ+PCQCQ\\
CPC+CQC &CPCQ+CQCQ
\end{pmatrix}=O.
\end{equation*}	
Comparing both sides, we obtain the following:
\[PCPC+PCQC=O\implies PC(P+Q)C=0\implies P=Q.\]
When $Q$ is substituted in (\ref{thm:3_eq:1}), we have $M$ in desired form as follows:
\[M=\begin{pmatrix}
PC &PCP\\
C &CP
\end{pmatrix}+I_{2n}.\]
Conversely, if $M$ has the following form:
\[M=\begin{pmatrix}
PC &PCP\\
C &CP
\end{pmatrix}+I_{2n}.\]
We have
\begin{eqnarray*}
&&M^2=\begin{pmatrix}
PC &PCP\\
C &CP
\end{pmatrix}^2+I_{2n} \\
&\implies& M^2= \begin{pmatrix}
PCPC+PCPC &PCPCP+PCPCP\\
CPC+CPC &CPCP+CPCP
\end{pmatrix}+I_{2n}\\
&\implies& M^2=I_{2n}.
\end{eqnarray*}

Therefore, if $M$ is in the form of (\ref{Eqn_IMDS_even}), then $M$ is involutory. This completes the proof.
\end{proof}

\noindent Since each sub-matrix of an MDS matrix is non-singular, according to Theorem~\ref{thm:3}, we can find all involutory MDS matrices of even order $2n$ by iterating over the entries of the matrices $P$ and $C$ of order $n$. Consequently, the search space for finding all involutory MDS matrices of even order $2n$ is reduced to $2^{2n^2m}$ over the field $\mathbb{F}_{2^m}$. Next, we will explore the relationship between our proposed matrix form and other matrix forms available in the literature for finding all involutory MDS matrices of even order.

\subsection{Relation to the Matrix Forms available in the literature}
In \cite{Yang2021}, the authors propose a matrix form to find all involutory MDS matrices of even order over $\mathbb{F}_{2^m}$. The proposed matrix form is given by:

\begin{equation*}
    \begin{aligned}
        M&=
        \begin{bmatrix}
            A_1 & A_2\\
            A_2^{-1}(I_n+A_1^2) & A_2^{-1}A_1A_2
        \end{bmatrix},
    \end{aligned}
\end{equation*}

where $A_1$ and $A_2$ are non-singular matrices over $\mathbb{F}_{2^m}$. Comparing our proposed matrix form given in~(\ref{Eqn_IMDS_even}), we have
\begin{equation*}
    \begin{aligned}
        & PC+I_n= A_1, && PCP= A_2\\
        & C= A_2^{-1}(I_n+A_1^2), && CP+I_n= A_2^{-1}A_1A_2.
    \end{aligned}
\end{equation*}
 
Note that since $M$ is MDS, we have $A_2^{-1}(I_n+A_1^2)$ is non-singular. Therefore, $(I_n+A_1^2)$ must also be non-singular. Now, $(I_n+A_1^2)=(I_n+A_1)^2$, implying that $(I_n+A_1)$ is also non-singular. Therefore, from the above equations, we can derive that
\begin{equation*}
    \begin{aligned}
        P=(I_n+A_1)^{-1}A_2~\text{and}~ C=A_2^{-1}(I_n+A_1^2).
    \end{aligned}
\end{equation*}
\\
\noindent The proposed matrix form in~\cite{Samanta2023} for generating all involutory MDS matrices of even order over $\mathbb{F}_{2^m}$ is given by:

\begin{equation*}
    \begin{aligned}
        M&=
        \begin{bmatrix}
            A_1 & (I_n+A_1^2)A_3^{-1}\\
            A_3 & A_3A_1A_3^{-1}
        \end{bmatrix},
    \end{aligned}
\end{equation*}
where $A_1$ and $A_3$ are non-singular matrices over $\mathbb{F}_{2^m}$. Comparing our proposed matrix form given in~(\ref{Eqn_IMDS_even}), we have
\begin{equation*}
    \begin{aligned}
        & PC+I_n= A_1, && PCP= (I_n+A_1^2)A_3^{-1}\\
        & C= A_3, && CP+I_n= A_3A_1A_3^{-1}.
    \end{aligned}
\end{equation*}

Note that since $M$ is MDS, $A_3$ is non-singular. Additionally, since $(I_n+A_1^2)A_3^{-1}$ is non-singular and $(I_n+A_1^2)=(I_n+A_1)^2$, implying that $(I_n+A_1)$ is also non-singular. Therefore, from the above equations, we can derive that
\begin{equation*}
    \begin{aligned}
        P=(I_n+A_1)A_3^{-1}~\text{and}~ C=A_3.
    \end{aligned}
\end{equation*}

\section{Construction of $4 \times 4$ Involutory MDS Matrices}~\label{Sec:proposed_method}
This section introduces a hybrid technique for generating all $4\times 4$ involutory MDS matrices over $\mathbb{F}_{2^m}$. There are $(2^m)^{16}$ matrices of order $4\times 4$ over $\F_{2^m}$ in existence, of which $((2^m)^4-1)((2^m)^4-2^m)((2^m)^4-(2^m)^2)((2^m)^4-(2^m)^3)$ are non-singular. It is not easy to filter out all $4\times 4$ involutory MDS matrices from the collection of all $4\times 4$ non-singular matrices. By resolving a small set of non-linear equations over $\F_{2^m}$, the authors in \cite{gmt} propose a matrix form to generate all $3 \times 3$ involutory MDS matrices over $\F_{2^m}$. However, it is not simple to generalize the concept of \cite{gmt} for orders of $4 \times 4$ or greater.

For a brute force search for involutory MDS matrices over $\mathbb{F}_{2^m}$, the search space is the entire set of non-singular matrices over $\mathbb{F}_{2^m}$. Some obvious checks may be imposed on the structure of the matrix for MDS search, such as choosing the entries of the matrices to be non-zero. However, this would not result in a significant reduction in the space and search time. Another idea for refining the search is to search over involutory matrices only. This would provide a good amount of reduction in search space but still offers a fairly large search space even for smaller matrix dimensions. For this purpose, we utilize a systematic construction that yields a search space much smaller than the set of all involutory matrices over $\mathbb {F}_{2^m}$.

The main goal of our suggested approach is to first locate all involutory MDS class representative matrices of order $4$ through search and then to obtain all $4 \times 4$ involutory MDS matrices by multiplying a non-singular diagonal matrix and its inverse with these involutory MDS class representative matrices. 

Let $\mathcal{M}_4$ denote the matrix space containing all $4 \times 4$ involutory MDS matrices. According to Lemma~\ref{lem:1}, if $M \in \mathcal{M}_4$, then for any non-singular diagonal matrix $D$, $D^{-1}MD \in \mathcal{M}_4$. Now, let's define an equivalence relation on $\mathcal{M}_4$, such that matrices $M$ and $M'$ are related if there exists an invertible diagonal matrix $D = \Diag(1, b_1, b_2, b_3)$ over $\mathbb{F}_{2^m}^{*}$ such that $M= D^{-1}M'D$. Thus, for $M \in \mathcal{M}_4$, if $\bar{M}$ denotes the equivalence class of $M$, all the matrices in $\bar{M}$ can be obtained by $D^{-1}MD$, varying the diagonal matrices $D = \text{Diag}(1, b_1, b_2, b_3)$ over $\mathbb{F}_{2^m}^{*}$.\\

To find the class representative matrices of all $4 \times 4$ involutory MDS matrices over $\mathbb{F}_{2^m}^*$, we define the matrix form $R$ as follows:

\begin{equation}\label{rep_mat_eqn}
R=\begin{pmatrix}PC &PCP\\C &CP\end{pmatrix}+I_4,
\end{equation}
where $I_4$ is a $4\times 4$ identity matrix, and $P$ and $C$ are the $2\times 2$ matrices given by

\begin{equation*}
    \begin{aligned}
        C&=
        c\begin{pmatrix}
            pq+r &p\\
            q &1
        \end{pmatrix} && 
        P=\begin{pmatrix}
            d+1 &d\\
            d &d+1
        \end{pmatrix},
    \end{aligned}
\end{equation*}
where $p,q,r,c,d\in\F_{2^m}^*$.

It can be observed that the sum of any rows or columns of $P$ is $1$. Additionally, the $2 \times 2$ non-singular matrices $C$ and $P$ over $\mathbb{F}_{2^m}^*$ define the matrix form $R$ mentioned above for determining all involutory MDS class representative matrices of order $4$. Consequently, by searching within this reduced subset of all $4 \times 4$ involutory matrices, the search space is now constrained to $(2^m-1)^5$.\\

\noindent Let $\mathcal{R}_4$ be the matrix space containing all representative involutory matrices of the equivalence classes. Thus, we can write the set $\mathcal{M}_4$ of all $4\times 4$ involutory MDS matrices as follows:

\[
\mathcal{M}_4= \underset{R \text { is MDS}} {\underset {R \in  \mathcal{R}_4} {\bigcup}}\{D^{-1}RD:~D=\Diag(1,b_1,b_2,b_3),~\text{where}~ b_1,b_2,b_3 \in \mathbb{F}_{2^m}^*\}.
\]

We can associate elements of $\mathcal{R}_4$ with the tuple $(p, q, r, c, d)$ because $\mathcal{R}_4$ consists of five variables: $p, q, r, c$, and $d$ in $\mathbb{F}_{2^m}^*$. Now, let's demonstrate that two different tuples correspond to two different elements of $\mathcal{R}_4$.
	
\begin{prp}
The cardinality of the set $\mathcal{R}_4$ is $(2^m-1)^5$.
\end{prp}

\begin{proof}
Consider the two tuples, $(p_1,q_1,r_1,c_1,d_1)$ and $(p_2,q_2,r_2,c_2,d_2)$, which are linked to the same elements $R \in \mathcal{R}_4$. Let $P_1, P_2, C_1, C_2$ as 
    \begin{eqnarray*}	
    &&C_1=c_1\begin{pmatrix}
    p_1q_1+r_1 &p_1\\
    q_1 &1
    \end{pmatrix},
    C_2=c_2\begin{pmatrix}
    p_2q_2+r_2 &p_2\\
    q_2 &1
    \end{pmatrix}\\
    &&P_1=\begin{pmatrix}
    d_1+1 &d_1\\
    d_1 &d_1+1
    \end{pmatrix}
    \text{ and }
    P_2=\begin{pmatrix}
    d_2+1 &d_2\\
    d_2 &d_2+1
    \end{pmatrix}.
    \end{eqnarray*} 

    Now, based on the assumption, we have
    \begin{eqnarray*}
    &&\begin{pmatrix}P_1C_1 &P_1C_1P_1\\C_1 &C_1P_1\end{pmatrix}+I_4=\begin{pmatrix}P_2C_2 &P_2C_2P_2\\C_2 &C_2P_2\end{pmatrix}+I_4\\
    &\implies &\begin{pmatrix}P_1C_1 &P_1C_1P_1\\C_1 &C_1P_1\end{pmatrix}=\begin{pmatrix}P_2C_2 &P_2C_2P_2\\C_2 &C_2P_2\end{pmatrix}\\
    &\implies& C_1=C_2 \text{ and } C_1P_1=C_2P_2\\
    &\implies&  c_1=c_2, c_1p_1=c_2p_2, c_1q_1=c_2q_2 ~\text{and}~ d_1=d_2\\
    &\implies& (p_1,q_1,r_1,c_1,d_1)=(p_2,q_2,r_2,c_2,d_2).
    \end{eqnarray*} 
This shows that distinct elements of $\mathcal{R}_4$ are linked to distinct tuples. Hence $|\mathcal{R}_4|=(2^m-1)^5$.
\end{proof}

\noindent In the following example, we illustrate our construction of a representative involutory MDS matrix $R$ of order $4$ over $\mathbb{F}_{2^4}$.  

\begin{example}
Let  $x^4 + x + 1$ be the primitive polynomial that generates $\mathbb{F}_{2^4}$. Let $\alpha$ be a primitive element that is a root of this polynomial. we define the matrix $R$ as follows:
\[R=
\begin{pmatrix}
PC &PCP\\
C &CP
\end{pmatrix}+I_4,\]
where  $C=
c\begin{pmatrix}
pq+r &p\\
q &1
\end{pmatrix}$ and
$P=\begin{pmatrix}
d+1 &d\\
d &d+1
\end{pmatrix}$ are $2\times 2$ non-singular matrices. \\
\noindent Let $p,q,r,c,d=1,~1, ~\alpha, ~\alpha, ~\alpha$. Thus, 
$C=\begin{pmatrix}
\alpha^5&\alpha\\
\alpha    &\alpha
\end{pmatrix}$ and
$P=\begin{pmatrix}
\alpha^4&\alpha\\
\alpha&\alpha^4
\end{pmatrix}$.
Hence, we have \\
	
\[PC=\begin{pmatrix}
\alpha^{11} & \alpha\\
\alpha^9  & \alpha\\
\end{pmatrix},
CP=\begin{pmatrix}
\alpha^{11} & \alpha^9\\
\alpha  & \alpha
\end{pmatrix}
\text{ and  }
PCP=\begin{pmatrix}
\alpha^8 & \alpha^{14}\\
\alpha^{14}  & 1
\end{pmatrix}.\] \\
We will now construct a representative involutory MDS matrix $R$ of order $4$.
\begin{eqnarray*}
R&=&\begin{pmatrix}
PC&PCP\\
C &CP
\end{pmatrix}+I_4\\
&=&\begin{pmatrix}
\alpha^{11}&\alpha&\alpha^{8}&\alpha^{14}\\
\alpha^9&\alpha&\alpha^{14}&1\\
\alpha^5&\alpha&\alpha^{11}&\alpha^9\\
\alpha&\alpha&\alpha&\alpha
\end{pmatrix}+I_4\\
&=&\begin{pmatrix}
\alpha^{12}&\alpha&\alpha^{8}&\alpha^{14}\\
\alpha^9&\alpha^4&\alpha^{14}&1\\
\alpha^5&\alpha&\alpha^{12}&\alpha^9\\
\alpha&\alpha&\alpha&\alpha^4
\end{pmatrix}.
\end{eqnarray*}

\end{example}

\begin{rem}
It is noteworthy that the authors of~\cite{Yang2021} demonstrate that the search space for finding involutory MDS matrices of even order over $\mathbb{F}_{2^m}$ may be reduced to $2^{(mn^2)/2}$, which is $2^{8m}$ when $n=4$. In contrast, the search space in our method is $(2^m-1)^5$, significantly smaller than that in~\cite{Yang2021}.
\end{rem}
    
\vspace*{2em}

\noindent Now, we will delve into the theoretical aspects behind the proposed matrix form given in~(\ref{rep_mat_eqn}). For this discussion, we require the following significant characteristics for involutory MDS matrices.

\begin{thm}\label{thm:4}
For any involutory MDS matrix $M$ in $\mathcal{L}_4$ written as follows:
    \[
        M=\begin{pmatrix}
            PC & PCP \\
            C & CP
        \end{pmatrix} + I_4,
    \]
the matrices $C$ and $P$ have all non-zero entries.
\end{thm}
\begin{proof}
Let $M=(m_{ij}), C=\begin{pmatrix}
c_{11} &c_{12}\\c_{21} &c_{22}
\end{pmatrix}$ and $P=\begin{pmatrix}
p_{11} &p_{12}\\p_{21} &p_{22}
\end{pmatrix}$.

Since $M$ is an MDS matrix, we can say that all the entries in $C$ are non-zero. Now, consider the following sub-matrices of $M$:
\[M_1=
\begin{pmatrix}
m_{12} &m_{14}\\m_{32} &m_{34}
\end{pmatrix},
M_2=
\begin{pmatrix}
m_{12} &m_{13}\\m_{42} &m_{43}
\end{pmatrix},
M_3=
\begin{pmatrix}
m_{21} &m_{24}\\m_{31} &m_{34}
\end{pmatrix},
M_4=
\begin{pmatrix}
m_{21} &m_{23}\\m_{41} &m_{43}
\end{pmatrix}.\]

These sub-matrices can be represented in terms of the elements of $C$ and $P$ as follows:

\begin{equation*}
    \begin{aligned}
        M_1 &=
        \begin{pmatrix}
            c_{12}p_{11} + c_{22}p_{12} &c_{11}p_{11}p_{12} + c_{21}p_{12}^2 + c_{12}p_{11}p_{22} + c_{22}p_{12}p_{22}\\
            c_{12} &c_{11}p_{12} + c_{12}p_{22}
        \end{pmatrix}\\
        M_2 &=
        \begin{pmatrix}
            c_{12}p_{11} + c_{22}p_{12}  &c_{11}p_{11}^2 + c_{21}p_{11}p_{12} + c_{12}p_{11}p_{21} + c_{22}p_{12}p_{21}\\
            c_{22}        &c_{21}p_{11} + c_{22}p_{21}
        \end{pmatrix}\\
        M_3 &=
        \begin{pmatrix}
            c_{11}p_{21} + c_{21}p_{22}  &c_{11}p_{12}p_{21} + c_{21}p_{12}p_{22} + c_{12}p_{21}p_{22} + c_{22}p_{22}^2\\
            c_{11}                    & c_{11}p_{12} + c_{12}p_{22}
        \end{pmatrix}\\
        M_4 &=  
        \begin{pmatrix}
            c_{11}p_{21} + c_{21}p_{22}  &c_{11}p_{11}p_{21} + c_{12}p_{21}^2 + c_{21}p_{11}p_{22} + c_{22}p_{21}p_{22}\\
            c_{21}                 & c_{21}p_{11} + c_{22}p_{21}
        \end{pmatrix}.
    \end{aligned}
\end{equation*}

The following are the determinants of these sub-matrices.
\begin{eqnarray*}
|M_1|=c_{12}c_{21}p_{12}^2 + c_{11}c_{22}p_{12}^2=p_{12}^2\det{(C)}\\
|M_2|=c_{12}c_{21}p_{11}^2 + c_{11}c_{22}p_{11}^2=p_{11}^2\det{(C)}\\
|M_3|=c_{12}c_{21}p_{22}^2 + c_{11}c_{22}p_{22}^2=p_{22}^2\det{(C)}\\
|M_4|=c_{12}c_{21}p_{21}^2 + c_{11}c_{22}p_{21}^2=p_{22}^2\det{(C)}.
\end{eqnarray*}
Given that $M$ is an MDS matrix, all the sub-matrices $M_1,M_2,M_3$ and $M_4$ are non-singular. Therefore, $p_{11}$, $p_{12}$, $p_{21}$, and $p_{22}$ must be non-zero.
\end{proof}

By Theorem~\ref{thm:4}, we know that neither the matrices $P$ nor $C$ can have zero entries, leading to the following result.

\begin{cor}\label{cor:1}
Let $R=\begin{pmatrix}
            PC & PCP \\
            C & CP
        \end{pmatrix} + I_4$ 
be an involutory MDS matrix in $\mathcal{R}_4$, where 
$C = c\begin{pmatrix}
    pq+r & p\\
    q & 1
\end{pmatrix}$ and  
$P = \begin{pmatrix}
    d+1 & d\\
    d & d+1
\end{pmatrix}$. Then, we have
\begin{enumerate}
    \item $r \neq pq$,
    \item $d \neq 1$.
\end{enumerate}
\end{cor}

\noindent The following theorem states that we can generate only one $4 \times 4$ involutory MDS matrix from a representative involutory MDS matrix of order $4$ by using the non-singular diagonal matrix $D=\Diag(1,b_1,b_2,b_3)$ over $\F_{2^m}^*$.

\begin{thm}\label{thm:5}
If a $4\times 4$ involutory MDS matrix $M \in [R]$, where $[R]$ denotes the equivalence class of the representative matrix $R \in \mathcal{R}_4$, then there exists a unique diagonal matrix $D = \Diag(1,b_1,b_2,b_3)$ over $\mathbb{F}_{2^m}^*$ such that $D^{-1}RD = M$.
\end{thm}

\begin{proof}
    Let $M$ be any involutory MDS matrix of order $4$ over $\F_{2^m}$ and suppose that there exists a non-singular diagonal matrix $D=\Diag(1,b_1,b_2,b_3)$ and a representative involutory MDS  matrix $R\in\mathcal{R}_4$ such that
    \begin{equation}\label{thm:5_eq:1}
    D^{-1}RD=M.
    \end{equation}
    Since $M$ is an involutory MDS matrix, by using Theorem \ref{thm:3}, it can be expressed as follows: 
    \[M=\begin{pmatrix}
    PC &PCP\\
    C &CP
    \end{pmatrix}+I_4,\]
    where $C,P$ are $2\times 2$ non-singular matrices.  Let 
    $C=\begin{pmatrix}
    c_{11}&c_{12}\\
    c_{21}&c_{22}\\
    \end{pmatrix}$ and 
    $P=\begin{pmatrix}
    p_{11}&p_{12}\\
    p_{21}&p_{22}\\
    \end{pmatrix}$.
    
    We can write diagonal matrix $D$ as 
    $\begin{pmatrix}
    B_1 &O\\
    O&B_2
    \end{pmatrix}$,
    where $B_1=\begin{pmatrix}
    1 &0\\
    0 &b_1
    \end{pmatrix}$ and
    $B_2=\begin{pmatrix}
    b_2 &0\\
    0 &b_3
    \end{pmatrix}$. 
    Given that $R\in\mathcal{R}_4$, so we can also write $R$ as follows: 
    \[R=\begin{pmatrix} P_1C_1 &P_1C_1P_1\\
     C_1 &C_1P_1 
     \end{pmatrix}+I_4, \]
    where $C_1$ and $P_1$ are still undetermined.\\
    
    \noindent Now since $DMD^{-1}=R$, we have
    \begin{equation*}
    \begin{pmatrix}
    B_1 &O\\
    O&B_2
    \end{pmatrix}
    \left[\begin{pmatrix}
    PC &PCP\\
    C &CP
    \end{pmatrix}+I_4\right]\begin{pmatrix}
    B_1^{-1} &O\\
    O&B_2^{-1}
    \end{pmatrix}=\begin{pmatrix}
    P_1C_1 &P_1C_1P_1\\
    C_1 &C_1P_1
    \end{pmatrix}+I_4
    \end{equation*}  
    or,
    \begin{equation*}
    \begin{pmatrix}
    B_1 &O\\
    O&B_2
    \end{pmatrix}
    \begin{pmatrix}
    PC &PCP\\
    C &CP
    \end{pmatrix}\begin{pmatrix}
    B_1^{-1} &O\\
    O&B_2^{-1}
    \end{pmatrix}=\begin{pmatrix}
    P_1C_1 &P_1C_1P_1\\
    C_1 &C_1P_1
    \end{pmatrix}\\
   \end{equation*}
    or,   
    \begin{equation}\label{thm:5_eq:2}
    \begin{pmatrix}
    B_1PCB_1^{-1} &B_1PCPB_2^{-1}\\
    B_2CB_1^{-1} &B_2CPB_2^{-1}
    \end{pmatrix}
    =\begin{pmatrix}
    P_1C_1 &P_1C_1P_1\\
    C_1 &C_1P_1
    \end{pmatrix}.
     \end{equation}
            
    We have from (\ref{thm:5_eq:2}) 
    \begin{eqnarray}
    B_1PCB_1^{-1}=P_1C_1\nonumber\\
    B_1PCPB_2^{-1}=P_1C_1P_1\nonumber\\
    B_2CPB_2^{-1}=C_1P_1\label{thm:5_eq:3}\\
    B_2CB_1^{-1}=C_1\label{thm:5_eq:4}.
    \end{eqnarray}
    
    Let $C_1=c\begin{pmatrix}
    s&p\\
    q &1
    \end{pmatrix}$. Then by using (\ref{thm:5_eq:4}), we obtain
    \begin{eqnarray*}
    s&=&b_2c_{11}c^{-1}\\
    p&=&b_2b_1^{-1}c_{12}c^{-1}\\
    q&=&b_3c_{21}c^{-1}\\
    c&=&b_3b_1^{-1}c_{22}.
    \end{eqnarray*}

    Since $c_{11}, c_{12}, c_{21}$, and $c_{22}$ are elements of the sub-matrix $C$ of the MDS matrix $M$, we can conclude that all of $c_{11}, c_{12}, c_{21}$, and $c_{22}$ are non-zero. Additionally, since $b_1, b_2$, and $b_3$ are non-zero, we can infer that $s, p, q$, and $c$ are also non-zero. Thus, we have

    \begin{equation}\label{Eqn_the_form_C1_in_Rep}
        \begin{aligned}
            C_1&=
            \begin{pmatrix}
                b_2c_{11}& b_2b_1^{-1}c_{12}\\
                b_3c_{21} & b_3b_1^{-1}c_{22}
            \end{pmatrix}.
        \end{aligned}
    \end{equation}
        
    Now, $|C_1|=b_2b_3b_1^{-1}c_{11}c_{22}+b_2b_3b_1^{-1}c_{21}c_{12}=b_2b_3b_1^{-1}|C|$, which indicates that $|C_1|\neq 0$. Since $C_1$ is non-singular, let 
    $\begin{vmatrix}
        s&p\\
        q &1
    \end{vmatrix}=r$, where $r\neq 0$. So $r=pq+s\implies s=pq+r$. Thus, we have 
    $C_1=c\begin{pmatrix}
        pq+r&p\\
        q &1
    \end{pmatrix}$. Also, according to our proposed representative matrix form, we can write the matrix $P_1$ as $P_1=\begin{pmatrix} d+1 &d\\ d &d+1 \end{pmatrix}$, where $d\in \mathbb{F}_{2^m}^{*}\setminus \set{1}$.

    \noindent Now, from (\ref{thm:5_eq:3}) and (\ref{thm:5_eq:4}), we have
    \begin{eqnarray}
    &&(B_2CB_1^{-1})^{-1}B_2CPB_2^{-1}=C_1^{-1}C_1P_1\nonumber\\
    &\implies& B_1C^{-1}B_2^{-1}B_2CPB_2^{-1}=C_1^{-1}C_1P_1\nonumber\\
    &\implies& B_1PB_2^{-1}=P_1\nonumber\\
    &\implies& \begin{pmatrix}
    1 &0\\
    0&b_1
    \end{pmatrix}
    \begin{pmatrix}
    p_{11} &p_{12}\\
    p_{21}&p_{22}
    \end{pmatrix}
    \begin{pmatrix}
    b_2^{-1} &0\\
    0&b_3^{-1}
    \end{pmatrix}=P_1\nonumber\\
    &\implies& 
    \begin{pmatrix}
    p_{11} &p_{12}\\
    b_1p_{21}&b_1p_{22}
    \end{pmatrix}
    \begin{pmatrix}
    b_2^{-1} &0\\
    0&b_3^{-1}
    \end{pmatrix}=P_1\nonumber\\
    &\implies& 
    P_1=\begin{pmatrix}
    p_{11}b_2^{-1} &p_{12}b_3^{-1}\\
    b_1p_{21}b_2^{-1}&b_1p_{22}b_3^{-1}
    \end{pmatrix}.\label{Eqn_form_P1_in_Rep}
    \end{eqnarray}

    Thus, from Equations (\ref{Eqn_the_form_C1_in_Rep}) and (\ref{Eqn_form_P1_in_Rep}), we observe that the matrices $P_1$ and $C_1$, i.e., the representative matrix $R$, can be expressed in terms of the entries of the sub-matrices $C$ and $P$ of the MDS matrix $M$.\\

    \noindent Now, we will determine the values $b_1,b_2$ and $b_3$ of the diagonal matrix $D$ in terms of the entries of the sub-matrices $C$ and $P$. Since the sum of each row or column of $P_1$ is equal to 1, we have
    \begin{eqnarray}
        &&p_{11}b_2^{-1}+p_{12}b_3^{-1}=1\label{thm:5_eq:5}\\
        &&b_1p_{21}b_2^{-1}+b_1p_{22}b_3^{-1}=1\label{thm:5_eq:6}\\
        &&p_{11}b_2^{-1}+b_1p_{21}b_2^{-1}=1\label{thm:5_eq:7}\\
        &&p_{12}b_3^{-1}+b_1p_{22}b_3^{-1}=1.\label{thm:5_eq:8}
    \end{eqnarray}

    Now, from (\ref{thm:5_eq:7}) and (\ref{thm:5_eq:8}), we have
    \begin{eqnarray}
    &&b_2=p_{11}+b_1p_{21}\label{thm:5_eq:9}\\
    &&b_3=p_{12}+b_1p_{22}\label{thm:5_eq:10}.
    \end{eqnarray}

    By solving the above equations, we have $b_1=\sqrt{\frac{p_{11}p_{12}}{p_{21}p_{22}}}$. Thus, $b_1$ has been uniquely determined by the terms of the entries of the sub-matrix $P$. We can also uniquely determine $b_2$ and $b_3$ by substituting the value of $b_1$ in (\ref{thm:5_eq:9}) and (\ref{thm:5_eq:10}). Thus, we have
    \begin{equation}\label{Eqn_form_Diag_matrix}
    \begin{aligned}
    b_1&= \sqrt{\frac{p_{11}p_{12}}{p_{21}p_{22}}},~b_2=p_{11}+p_{21}\sqrt{\frac{p_{11}p_{12}}{p_{21}p_{22}}},~b_3=p_{12}+p_{22}\sqrt{\frac{p_{11}p_{12}}{p_{21}p_{22}}}.
    \end{aligned}
    \end{equation}

    Now, according to Theorem \ref{thm:4}, we have $p_{11}, p_{12}, p_{21}, p_{22} \in \mathbb{F}_{2^m}^{*}$ which implies $b_1 \in \mathbb{F}_{2^m}^{*}$. 

    If possible let $b_2=0$. This implies that
    \begin{equation*}
        \begin{aligned}
            & p_{11}+b_1p_{21}=0\\
            \implies & b_1=\frac{p_{11}}{p_{21}}\\
            \implies & \sqrt{\frac{p_{11}p_{12}}{p_{21}p_{22}}} = \frac{p_{11}}{p_{21}}\\
            \implies & p_{11}p_{22}+p_{12}p_{21}=0.
        \end{aligned}
    \end{equation*}

    This leads to a contradiction since $P$ is non-singular. Therefore, $b_2 \in \mathbb{F}_{2^m}^{*}$. Similarly, we can demonstrate that $b_3 \in \mathbb{F}_{2^m}^{*}$. Consequently, we can uniquely determine the values $b_1, b_2$, and $b_3$ of the diagonal matrix $D$. In other words, we can conclude that there exists a unique diagonal matrix $D = \text{Diag}(1, b_1, b_2, b_3)$ over $\mathbb{F}_{2^m}^{*}$ such that $D^{-1}RD = M$.
\end{proof}

\noindent In the following example of an involutory MDS matrix $M$, taken from the literature~\cite[Example~4]{psa}, we illustrate the representative involutory MDS matrix $R$ and the non-singular diagonal matrix $D$, which are used to create the involutory MDS matrix $M$.

\begin{example}
Let  $x^4 + x + 1$ be the primitive polynomial that generates $\F_{2^4}$. Let $\alpha$ be a primitive element that is a root of this polynomial. Now, consider the $4 \times 4$ involutory MDS matrix $M$ over $\mathbb{F}_{2^4}$.

\[M=\begin{pmatrix}
 1   & 1 & 1 & 1 \\
 1   & \alpha  & \alpha^2  & \alpha^5 \\
\alpha^7 & \alpha^{10}  & \alpha & \alpha^5\\
\alpha^9  & \alpha^2 & \alpha^{10}   & 1
\end{pmatrix}.\]
Theorem \ref{thm:3} allows $M$ to be expressed as 
 \[M=
 \begin{pmatrix}
 PC &PCP\\
 C &CP
\end{pmatrix}+I_4,\]

where  $C$ and $P$ are $2\times 2$ non-singular matrices. Let
$C=\begin{pmatrix}
c_{11}&c_{12}\\
c_{21}&c_{22}
\end{pmatrix}$ and
$P=\begin{pmatrix}
p_{11}&p_{12}\\
p_{21}&p_{22}
\end{pmatrix}$.\\

Here $C=\begin{pmatrix}
\alpha^7 & \alpha^{10}\\
\alpha^9  & \alpha^2\\
\end{pmatrix}$
and 
$CP=\begin{pmatrix}
	\alpha+1 & \alpha^5\\
	\alpha^{10}   & 0
    \end{pmatrix}
	=\begin{pmatrix}
     \alpha^4 & \alpha^5\\
     \alpha^{10}   & 0
     \end{pmatrix}.$
 Then, $P=
\begin{pmatrix}
\alpha^{10} & \alpha^{8}\\
1   & 1
\end{pmatrix}$.\\

From~(\ref{Eqn_form_Diag_matrix}), we can find the unique diagonal matrix $D=\Diag(1,b_1,b_2,b_3)$ as  
\begin{eqnarray*}
b_1&=&\sqrt{\frac{p_{11}p_{12}}{p_{21}p_{22}}}=\sqrt {\alpha^{10}.\alpha^8}=\alpha^{9}\\
b_2&=&p_{11}+b_1p_{21}=\alpha^{10}+\alpha^{9}=\alpha^{13}\\
b_3&=&p_{12}+b_1p_{22}=\alpha^{8}+\alpha^{9}=\alpha^{12}.
\end{eqnarray*}

We will now locate a representative involutory MDS matrix $R$ of the class, to which the involutory MDS matrix $M$ belongs.
Theorem \ref{thm:3} also allows $R$ to be expressed as 
 \[R=\begin{pmatrix}
P_1C_1 &P_1C_1P_1\\
C_1 &C_1P_1
\end{pmatrix}+I_4.\]

From (\ref{Eqn_the_form_C1_in_Rep}) and (\ref{Eqn_form_P1_in_Rep}), we can find $C_1$ and $P_1$ as follows: 
\[C_1= \begin{pmatrix}
b_2c_{11} & b_2b_1^{-1}c_{12}\\
b_3c_{21} & b_3b_1^{-1}c_{22}
\end{pmatrix}=\begin{pmatrix}
\alpha^5 & \alpha^{14}\\
\alpha^6 & \alpha^5
\end{pmatrix}
,\text{and }
P_1=\begin{pmatrix}
p_{11}b_2^{-1} &p_{12}b_3^{-1}\\
b_1p_{21}b_2^{-1}&b_1p_{22}b_3^{-1}
\end{pmatrix}
=\begin{pmatrix}
\alpha^{12} &\alpha^{11}\\
\alpha^{11} &\alpha^{12}
\end{pmatrix}.\]
Now, we get
\[P_1C_1=\begin{pmatrix}
0 & \alpha^6\\
\alpha^9 & \alpha^4
\end{pmatrix}
, C_1P_1=\begin{pmatrix}
\alpha^4 & \alpha^6\\
\alpha^9 & 0
\end{pmatrix}
,\text{and }
P_1C_1P_1=\begin{pmatrix}
\alpha^2 &\alpha^3\\
\alpha^{13} &\alpha^2
\end{pmatrix}.\] 
Hence, we have 
\[R=\begin{pmatrix}
        P_1C_1 &P_1C_1P_1\\
        C_1 &C_1P_1
    \end{pmatrix}+I_4
    =
    \begin{pmatrix}
        1   & \alpha^6 & \alpha^2 & \alpha^3 \\
        \alpha^9   & \alpha  & \alpha^{13}  & \alpha^2 \\
        \alpha^5 & \alpha^{14}  & \alpha & \alpha^6\\
        \alpha^6  & \alpha^5 & \alpha^9   & 1
    \end{pmatrix}.
\]
\end{example}

\noindent In~\cite{Tuncay23}, the authors propose a representative matrix form to generate all involutory MDS matrices of order $4$. In their representative matrix, the sum of any rows (and columns) equals $1$. In the following theorem, we prove that these conditions are also guaranteed for our representative form.

\begin{thm}\label{thm:6}
The sum of any rows (and columns) in the involutory matrix $M$ in $\mathcal{L}_{2n}$ written as follows:
\begin{equation*}
M=\begin{pmatrix}
PC &PCP\\
C &CP
\end{pmatrix}+I_{2n} 
\end{equation*}
is equal to 1 if and only if the sum of any rows (and columns) of $P$ is equal to 1.
\end{thm}

\begin{proof}
Let $M=N+I_{2n}$, where 
$N=\begin{pmatrix}
PC &PCP\\
C &CP
\end{pmatrix}$. Assume that the sum of any rows (and columns) in $M$ equals 1. Let $J$ be a $n$-dimensional column vector with all entries 1. Then, we have
\begin{eqnarray*}
&&M
\begin{pmatrix}
J\\
J
\end{pmatrix}=
\begin{pmatrix}
J\\
J
\end{pmatrix}\\
&\implies &(N+I_{2n})
\begin{pmatrix}
J\\
J
\end{pmatrix}=
\begin{pmatrix}
J\\
J
\end{pmatrix}\\
&\implies &N
\begin{pmatrix}
J\\
J
\end{pmatrix}=
\begin{pmatrix}
O\\
O
\end{pmatrix}\\
&\implies &\begin{pmatrix}
PC &PCP\\
C &CP
\end{pmatrix}
\begin{pmatrix}
J\\
J
\end{pmatrix}=
\begin{pmatrix}
O\\
O
\end{pmatrix},
\end{eqnarray*}
where $O$ is the $n$-dimensional zero-column vector. Similar to this, we have
 \[\begin{pmatrix}
J^T&J^T\\
\end{pmatrix}\begin{pmatrix}
PC &PCP\\
C &CP
\end{pmatrix}
=\begin{pmatrix}
O^T &O^T\\
\end{pmatrix}.\]	
The matrix equations above lead us to the conclusion that 
\begin{eqnarray*}
&&(PC+PCP)J=O \\
&\implies& (I_n+P)J=O ~\hspace*{2em} [\text{since P and C are non-singular matrices}]\\
&\implies& PJ=J.
\end{eqnarray*} 

Likewise, we can show that $J^TP=J^T$. This demonstrates that the sum of any rows (and columns) in $P$ is equal to 1. \\

Conversely, suppose that the sum of any rows (and columns) in $P$ equals 1. Then, we have 
\begin{eqnarray*}
&&M
\begin{pmatrix}
J\\
J
\end{pmatrix} \\
&=&(N+I_{2n})
\begin{pmatrix}
J\\
J
\end{pmatrix}\\
&=&N
\begin{pmatrix}
J\\
J
\end{pmatrix}+ \begin{pmatrix}
J\\
J
\end{pmatrix}  \\
&=& \begin{pmatrix}
PC &PCP\\
C &CP
\end{pmatrix}
\begin{pmatrix}
J\\
J
\end{pmatrix}+\begin{pmatrix}
J\\
J
\end{pmatrix} \\
&=&\begin{pmatrix}
(PC+PCP)J\\
(C+CP)J
\end{pmatrix}
+\begin{pmatrix}
J\\
J
\end{pmatrix}\\
&=&\begin{pmatrix}
PC(I_n+P)J\\
C(I_n+P)J
\end{pmatrix}
+\begin{pmatrix}
J\\
J
\end{pmatrix} \\
&=&\begin{pmatrix}
PC(J+PJ)\\
C(J+PJ)
\end{pmatrix}
+\begin{pmatrix}
J\\
J
\end{pmatrix}\\
&=& \begin{pmatrix}
O\\
O
\end{pmatrix}
+\begin{pmatrix}
J\\
J
\end{pmatrix}~\hspace*{2em} [\text {Since } PJ=J]\\
&=& \begin{pmatrix}
J\\
J
\end{pmatrix}.
\end{eqnarray*} 

Similarly, we can demonstrate that 
$\begin{pmatrix}
J^T&J^T
\end{pmatrix}M=\begin{pmatrix}
J^T& J^T
\end{pmatrix}$.
Hence, the sum of any rows (and columns) in $M$ equals 1.
\end{proof}

\section{Experimental Results}~\label{Sec:Experimental_result}
This section presents our algorithm for generating all $4\times 4$ involutory MDS matrices over $\F_{2^m}$.  Now, let us consider the collection of diagonal matrices of order $4$ as follows: 

$$\mathcal{D}_4=\{\Diag(1,b_1,b_2,b_3)\mid b_1, b_2, b_3 \in \F_{2^m}^*\}.$$

\noindent We follow the steps below:

\begin{enumerate}
\item First, we generate all involutory MDS class representative matrices of order $4$.
\item From these involutory MDS class representative matrices, we generate all $4\times 4$ involutory MDS matrices.
\end{enumerate}

\begin{rem}
It should be noted that while generating an involutory MDS matrix, using the involutory matrix form $R$ offers an additional benefit. Since $R$ is always involutory, we have $|R| = 1$. Also, we know that if entries of the inverse (here $R^{-1}=R$) of a $4 \times 4$ non-singular matrix are non-zero, then all its $3 \times 3$ sub-matrices are non-singular. Thus, we need to check only sub-matrices of $R$ of order $1$ and $2$ to verify the MDS property of $R$.
\end{rem}

\begin{algorithm}[htpb!]
\caption{Printing all $4\times 4$  involutory MDS matrices over $\F_{2^m}$}
{\bf Input:} Finite Field $\F_{2^m}$.\\
{\bf Output:} All Involutory MDS matrices of order $4\times4$. 
\begin{algorithmic}[1]
\For{$p,q,r,c\in \F_{2^m}^*$ and $d\in \F_{2^m}^* \setminus \set{1}$}
\State $C\gets c\begin{pmatrix}
pq+r & p\\
q & 1
\end{pmatrix}$
and $P\gets \begin{pmatrix}
d+1 & d\\
d & d+1
\end{pmatrix}$
\State $R\gets \begin{pmatrix}
PC & PCP\\
C & CP
\end{pmatrix}+I_4$
\If{$R$ is not MDS} 
\State continue
\Else
\For{$D \in \mathcal{D}_4$ }
\State Print matrix $D^{-1}RD$
\EndFor
\EndIf
\EndFor
\end{algorithmic}
\end{algorithm}

Now, we generate all $4\times 4$ involutory MDS matrices over $\F_{2^m}$ for $m=3,4,5,6,7,8$ by searching through all involutory MDS class representative matrices over these finite fields. The following table gives the experimental findings. 

\noindent It is worth noting that in~\cite{Kumar_MDS2024}, the authors demonstrate a hybrid approach for generating all MDS matrices, resulting in a search space of $2^{9m}$ to generate all $4\times 4$ involutory MDS matrices. In another paper~\cite{Yang2021}, the search space for generating all $4\times 4$ involutory MDS matrices is $2^{8m}$. Additionally, a recent study~\cite{Tuncay23} proposes a hybrid method to generate all $4\times 4$ involutory MDS matrices over $\mathbb{F}_{2^m}$, reducing the search space to $(2^m-1)^8$ to identify all $4\times 4$ representative involutory MDS matrices. However, the enumeration of all $4\times 4$ involutory MDS matrices in these papers is limited to finite fields of small order; specifically, they can find the count $\mathbb{F}_{2^m}$ for $m\leq 4$. In contrast, in our proposed method, the search space for finding all representative matrices is $(2^m-1)^5$. Consequently, we can explicitly enumerate the total number of $4 \times 4$ involutory MDS matrices over $\mathbb{F}_{2^m}$ for $m=3,4,\ldots,8$.

\begin{center}
\begin{tabular}{|c|c|c|c|}\hline
Finite Field & Number of class representatives& Choices for $D$ & Total Count \\ \hline
$\F_{2^3}$ & $48$       & $7^3$  &$7^3\times48$\\   
$\F_{2^4}$ & $71856$    & $15^3$ &$15^3\times71856$\\ 
$\F_{2^5}$ & $10188240$    & $31^3$ &$31^3\times10188240$\\ 
$\F_{2^6}$ & $612203760$    & $63^3$ &$63^3\times612203760$\\ 
$\F_{2^7}$ & $26149708368$    & $127^3$ &$127^3\times26149708368$\\ 
$\F_{2^8}$ & $961006331376$    & $255^3$ &$255^3\times961006331376$\\ 
\hline
\end{tabular}
\vspace{2mm}\\
\captionof{table}{Count  of all $4\times 4$ involutory MDS matrices over $\F_{2^m}$}
\label{Table_4_IMDS}
\end{center}

\section{Conclusion}~\label{Sec:Conclusion}
In this article, we have presented some characterizations of involutory MDS matrices of even order. Our contributions include introducing a new matrix form for obtaining all involutory MDS matrices of even order and comparing it with other matrix forms available in the literature. Furthermore, we have proposed a technique to systematically construct all $4 \times 4$ involutory MDS matrices over a finite field $\mathbb{F}_{2^m}$. This method significantly reduces the search space by focusing on involutory MDS class representative matrices. Specifically, our approach involves searching for these representative matrices within a set of cardinality $(2^m-1)^5$. Through this method, we provide an explicit enumeration of the total number of $4 \times 4$ involutory MDS matrices over $\mathbb{F}_{2^m}$ for $m=3,4,\ldots,8$. This work may be expanded to suggest the generation of all $n \times n$ involutory MDS matrices over $\mathbb{F}_{2^m}$, where $n > 4$.\\

\noindent \textbf{Declarations}\\

\noindent \textbf{Conflict of interest} The authors declare that they have no conflict of interest.

\medskip

\bibliographystyle{plain}
\bibliography{4_IMDS_Ref}

\end{document}